\documentclass[cleveref]{lipics-v2021}

\usepackage[color=blue!30,colorinlistoftodos]{todonotes}

\pdfoutput=1 
\hideLIPIcs  
\nolinenumbers

\supplementdetails{Source code repository}{https://github.com/SPbSAT/simplifier} 

\usepackage{booktabs}

\usepackage{tikz}
\usetikzlibrary{shapes.geometric, calc}
\usetikzlibrary{positioning}

\tikzstyle{input}           = [inner sep=1pt]
\tikzstyle{gate}            = [circle,fill=white,draw=black,minimum
size=4mm,inner sep=0pt]
\tikzstyle{outgate}         = [gate, thick]
\tikzstyle{wire}            = [draw,->]
\tikzstyle{notwire}         = [draw,->,dashed]

\DeclareMathOperator{\SUM}{SUM}
\DeclareMathOperator{\AtLeast}{AtLeast}
\DeclareMathOperator{\AtMost}{AtMost}

\title{Simplifier: A~New Tool for Boolean Circuit Simplification}

\author{Daniil Averkov}{St.~Petersburg State University}{averkov.dan@gmail.com}{}{}

\author{Gregory Emdin}{École Polytechnique Fédérale de Lausanne}{grigorii.emdin@epfl.ch}{}{}

\author{Viktoriia Krivogornitsyna}{Steklov Mathematical Institute at St.~Petersburg, Russian Academy of Sciences}{krivogornitsyna.va@gmail.com}{https://orcid.org/0009-0003-0702-9973}{}

\author{Alexander~S. Kulikov}{JetBrains Research\and \url{https://alexanderskulikov.github.io}}{alexander.s.kulikov@gmail.com}{https://orcid.org/0000-0002-5656-0336}{}

\author{Fedor Kurmazov}{Steklov Mathematical Institute at St.~Petersburg, Russian Academy of Sciences}{f.kurmazov.b@gmail.com}{https://orcid.org/0009-0004-2072-3457}{}

\author{Alexander Smal}{JetBrains Research}{avsmal@gmail.com}{https://orcid.org/0000-0002-8241-5503}{}

\author{Vsevolod Vaskin}{Neapolis University Pafos}{vaskin.2003@gmail.com}{https://orcid.org/0009-0006-5187-8384}{}

\authorrunning{D.~Averkov, G.~Emdin, V. Krivogornitsyna, A.~Kulikov, F.~Kurmazov, A.~Smal, V.~Vaskin}

\Copyright{Daniil Averkov, Gregory Emdin, Viktoriia Krivogornitsyna, Alexander~S. Kulikov, Fedor Kurmazov, Alexander Smal, Vsevolod Vaskin}


\begin{document}
\sloppy

\maketitle

\begin{abstract}
    The Boolean circuit simplification problem involves finding a smaller circuit that computes the same function as~a~given Boolean circuit. This problem is closely related to several key areas with both theoretical and practical applications, such as~logic synthesis, satisfiability, and verification.

    In this paper, we present \texttt{Simplifier}, a~new open source tool for simplifying Boolean circuits. The tool optimizes subcircuits with three inputs and at most three outputs, seeking to improve each one. It is designed as a~low-effort method that runs in~just a~few seconds for circuits of~reasonable size. This efficiency is~achieved by combining two key strategies. First, the tool utilizes a~precomputed database of~optimized circuits, generated with SAT solvers after carefully clustering Boolean functions with three inputs and up~to three outputs. Second, we~demonstrate that it~is sufficient to~check
    a~linear number of~subcircuits, relative to~the size of~the original circuit. This allows a~single iteration of~the tool to~be executed in~linear time.

    We evaluated the tool on a~wide range of Boolean circuits, including both industrial and hand-crafted examples, in two popular formats: AIG and BENCH. For AIG circuits, after applying the state-of-the-art \texttt{ABC} framework, our tool achieved an additional 4\% average reduction in~size. For BENCH circuits, the tool reduced their size by~an~average of 30\%.
\end{abstract}


\section{Overview}

Boolean circuits as a~way to~represent Boolean functions appears in~many
applications, both in~theory and in~practice.
In~most cases, circuits are succinct and efficient descriptions
of~Boolean functions as~the size of a~circuit might be~much smaller than the size of
its truth-table (that is always $2^n$ for an $n$-variate Boolean function) or~a~CNF (that may also have exponential size).
For example,
the minimum number of~clauses in a~CNF formula that computes
the parity function $x_1 \oplus \dotsb \oplus x_n$
is~$2^{n-1}$, whereas this function can be~computed
by~a~circuit of~size~$n-1$ over the full binary basis $B_2$ and
a~circuit of~size $3(n-1)$ over the basis $\{\neg, \lor\}$ if one
allows negations on~the wires (see \Cref{fig:parity-5}).

\begin{figure}[ht]
\centering
\begin{tikzpicture}[xscale=0.8]
\foreach \a in {0,1} {
\foreach \b in {0,1} {
\foreach \c in {0,1} {
\foreach \d in {0,1} {
\foreach \e in {0,1} {
\draw let \n1={int(mod(\a+\b+\c+\d+\e,2))} in
	node[label=above:\rotatebox{90}{\tt \a\b\c\d\e}, minimum size=4mm, inner sep=0pt, draw, rectangle]
	at (8*\a + 4*\b + 2*\c + \d + 0.5*\e,0) {\tt \n1};
}}}}}

\node at (8, 1.7) {truth-table:};
\end{tikzpicture}
\vspace{3mm}

\begin{tikzpicture}
    \node[above] at (0,0) {CNF formula:};
    \node[text width=12cm, below] at (0, 0) {$({x_1} \lor {x_2} \lor {x_3} \lor {x_4} \lor {x_5}) \land({x_1} \lor {x_2} \lor {x_3} \lor \overline{x_4} \lor \overline{x_5}) \land({x_1} \lor {x_2} \lor \overline{x_3} \lor {x_4} \lor \overline{x_5}) \land({x_1} \lor {x_2} \lor \overline{x_3} \lor \overline{x_4} \lor {x_5}) \land({x_1} \lor \overline{x_2} \lor {x_3} \lor {x_4} \lor \overline{x_5}) \land({x_1} \lor \overline{x_2} \lor {x_3} \lor \overline{x_4} \lor {x_5}) \land({x_1} \lor \overline{x_2} \lor \overline{x_3} \lor {x_4} \lor {x_5}) \land({x_1} \lor \overline{x_2} \lor \overline{x_3} \lor \overline{x_4} \lor \overline{x_5}) \land(\overline{x_1} \lor {x_2} \lor {x_3} \lor {x_4} \lor \overline{x_5}) \land(\overline{x_1} \lor {x_2} \lor {x_3} \lor \overline{x_4} \lor {x_5}) \land(\overline{x_1} \lor {x_2} \lor \overline{x_3} \lor {x_4} \lor {x_5}) \land(\overline{x_1} \lor {x_2} \lor \overline{x_3} \lor \overline{x_4} \lor \overline{x_5}) \land(\overline{x_1} \lor \overline{x_2} \lor {x_3} \lor {x_4} \lor {x_5}) \land(\overline{x_1} \lor \overline{x_2} \lor {x_3} \lor \overline{x_4} \lor \overline{x_5}) \land(\overline{x_1} \lor \overline{x_2} \lor \overline{x_3} \lor {x_4} \lor \overline{x_5}) \land(\overline{x_1} \lor \overline{x_2} \lor \overline{x_3} \lor \overline{x_4} \lor {x_5})$};
\end{tikzpicture}

\begin{tikzpicture}[xscale=0.8, yscale=0.8]
\begin{scope}
\foreach \v/\x in {1,...,4}
	\node[input] (x\v) at (\x-1,3) {$x_\v$};

\node[input] (x5) at (4, 3) {$x_5$};

\foreach \v/\x/\y in { a/0.5/2, b/2.5/2, c/1.5/1}
	\node[gate] (\v) at (\x,\y) {$\oplus$};

\node[outgate] (d) at (2.5,0) {$\oplus$};

\foreach \s/\t in {x1/a, x2/a, x3/b, x4/b, a/c, b/c, c/d, x5/d}
    \draw[wire] (\s) -- (\t);

\node at (2,3.7) {circuit over~$B_2$:};
\end{scope}

\begin{scope}[shift={(8,-1.2)},yscale=0.7]
\foreach \v/\x in {1,...,4}
	\node[input] (x\v) at (\x-1,6) {$x_\v$};

\node[input] (x5) at (4, 6) {$x_5$};

\foreach \v/\x/\y/\a/\b in {a/0/4/x1/x2, b/2/4/x3/x4, c/1/2/a3/b3, d/2/0/c3/x5}
{
	\node[gate] (\v1) at (\x,\y+1) {$\land$};
	\node[gate] (\v2) at (\x+1,\y+1) {$\land$};
	\node[gate] (\v3) at (\x+.5,\y) {$\land$};
	\foreach \s/\t in {\v1/\v3, \v2/\v3, \a/\v2, \b/\v2}
    	\draw[notwire] (\s) -- (\t);
	\foreach \s/\t in {\a/\v1, \b/\v1}
    	\draw[wire] (\s) -- (\t);
};

\node[outgate] (d) at (2.5,0) {$\land$};

\node at (2,7) {circuit over~$\{\land, \neg\}$:};
\end{scope}
\end{tikzpicture}

\caption{Four representations of~the function $x_1 \oplus x_2 \oplus x_3 \oplus x_4 \oplus x_5$: a~truth table,
    a~CNF, a~circuit over the full binary basis~$B_2$ of~size~$4$, and     a~circuit over~$\{\neg, \land\}$ of~size~$12$ (for the circuits, the dashed wires are negated, the output gate is~shown in~bold).
    Each~of the shown representations is~known to~be optimal with respect
    to~its size.}\label{fig:parity-5}
\end{figure}

\subsection{Circuit Synthesis}

In~the \emph{circuit synthesis} problem,
one is~given a~specification of~a~Boolean function
and needs to~design a~small circuit computing this function.
The output of~computer-aided design systems (that solve this problem)
is often used in the manufacture of integrated circuits.
The smaller the circuit, the easier and cheaper it is to produce: it has less functional elements, requires less wire crossings, etc.
Whereas in~practice various heuristic approaches are used,
the circuit synthesis problem is also important from
the theoretical point of~view. There, it is known
as the \emph{minimum circuit size problem (MCSP)}:
given a~truth-table of~a~function, output a~circuit of the minimum size
computing this function.
It is one of the central problems studied in meta-complexity theory~\cite{kabanets2000circuit}.
MCSP is conjectured to not have a polynomial-time solution (note that we measure its complexity with respect to the truth-table size, i.e., it is not solvable in time $2^{O(n)}$).
At the same time it is still an open question whether the general case of MCSP is NP-complete. Some researchers suggest
that it might be NP-intermediate (neither in P nor NP-complete).
The special cases of this problem for Boolean formulas in conjunctive normal form (CNF) or in disjunctive normal form (DNF) are well studied. These special cases of MCSP are known to be NP-hard~\cite{ilango2020constant}.
Despite that, DNF-MCSP and CNF-MCSP are often solved in practice
for small number of input variables using the widely known Quine--McCluskey algorithm~\cite{quine1952problem,quine1955way,mccluskey1956minimization}.
There is also an alternative approach based on Binary Decision Diagrams (BDDs)~\cite{coudert2002twolevel}.
However, Boolean functions that arise in practice, e.g., in VLSI design, have hundreds or even thousands of input variables and for it there is no hope of finding provably minimal circuits.

\subsection{Circuit Simplification}
In~this paper, we~focus on~the \emph{circuit simplification} problem:
given a~Boolean circuit, find a~circuit of~smaller size computing
the same function. It~is a~difficult problem as~it~generalizes the circuit satisfiability problem:
checking whether a~given circuit is~unsatisfiable (that~is, whether it~computes the same~as a~trivial
circuit that always outputs~$0$) is a~special case of~circuit simplification. In~practice, before checking whether a~given
circuit is~satisfiable or~not, one usually first simplifies~it.
The simplification rules are also applied in~the process of~checking satisfiability (say, after assigning some inputs
of~the circuit). In~applications like this one, it~is crucial
that a~low-effort simplification method is~used.

Circuit simplification is~also a~close relative of~circuit synthesis.
On the one hand,
one can use circuit simplification as a~part of~circuit synthesis
as~well~as in~postprocessing stages.
On~the other hand, one can even use circuit simplification directly for circuit synthesis:
given a~specification of a~Boolean function, synthesize a~naive circuit for~it, then simplify~it.

\Cref{fig:simplification-examples} shows two toy examples 
of~circuit simplification. In~the first example, the $\oplus$-gate has no~outgoing edges (and it~is not an~output gate), thus it~is redundant and can be~removed from the circuit. In~the second example, there are two gates that compute $x_2\oplus x_3$, and so~they can be~merged into one gate.
Such simplification rules are similar to~basic simplification rules like unit clause and pure literal elimination used 
by~CNF SAT solvers.

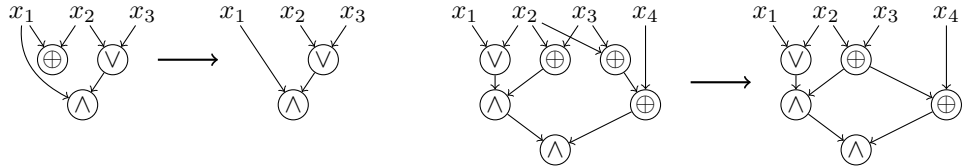
\begin{figure}[ht]
\centering
\begin{tikzpicture}[xscale=0.8, yscale=0.6]
\begin{scope}
\foreach \v/\x in {1,...,3}
	\node[input] (x\v) at (\x-1,2) {$x_\v$};

\foreach \v/\x/\y/\l in { a/0.5/1/\oplus, b/1.5/1/\lor, c/1/0/\land}
	\node[gate] (\v) at (\x,\y) {$\l$};

\foreach \s/\t in {x1/a, x2/a, x2/b, x3/b, b/c}
    \draw[wire] (\s) -- (\t);

\draw[wire] (x1) to [bend right=30] (c);
\end{scope}

\begin{scope}[xshift=3.5cm]
\foreach \v/\x in {1,...,3}
	\node[input] (x\v) at (\x-1,2) {$x_\v$};

\foreach \v/\x/\y/\l in { b/1.5/1/\lor, c/1/0/\land}
	\node[gate] (\v) at (\x,\y) {$\l$};

\foreach \s/\t in {x1/c, x2/b, x3/b, b/c}
    \draw[wire] (\s) -- (\t);

\end{scope}
\draw[wire, thick] (2.25,1) -- (3.25,1);
\node[] at (2.75,-1.5) {}; 
\end{tikzpicture}\quad\qquad
\begin{tikzpicture}[xscale=0.8, yscale=0.6]
\begin{scope}
\foreach \v/\x in {1,...,4}
	\node[input] (x\v) at (\x-1,3) {$x_\v$};

\foreach \v/\x/\y/\l in {a/0.5/2/\lor, b/1.5/2/\oplus, c/2.5/2/\oplus,
						 d/0.5/1/\land, e/3/1/\oplus, f/1.5/0/\land}
	\node[gate] (\v) at (\x,\y) {$\l$};

\foreach \s/\t in {x1/a, x2/a, x2/b, x3/b, x2/c,
				  x3/c, x4/e, c/e, a/d, b/d, d/f, e/f}
    \draw[wire] (\s) -- (\t);

\end{scope}

\begin{scope}[xshift=5cm]
\foreach \v/\x in {1,...,4}
	\node[input] (x\v) at (\x-1,3) {$x_\v$};

\foreach \v/\x/\y/\l in {a/0.5/2/\lor, b/1.5/2/\oplus,
						 d/0.5/1/\land, e/3/1/\oplus, f/1.5/0/\land}
	\node[gate] (\v) at (\x,\y) {$\l$};

\foreach \s/\t in {x1/a, x2/a, x2/b, x3/b, x4/e, b/e, a/d, b/d, d/f, e/f}
    \draw[wire] (\s) -- (\t);

\end{scope}

\draw[wire, thick] (3.75,1.5) -- (4.75,1.5);
\node[] at (4.25,-.5) {}; 
\end{tikzpicture}
\caption{Examples of circuit simplification: removing dangling gates (left) and merging duplicate gates (right).}
\label{fig:simplification-examples}
\end{figure}

To~give a~more involved example, consider the task of~synthesizing a~circuit computing the binary representation of~the sum of~three input bits. That~is, for three inputs bits $x_1,x_2,x_3$, the circuit needs to~output two bits $s$~and~$c$
(sum and carry bits) such that $x_1 + x_2 + x_3 = 2c+s$. Then, $s$~is equal to~$x_1+x_2+x_3$ modulo~$2$, whereas $c$~is the majority function of~$x_1,x_2,x_3$, that~is, $c$~is equal to~$1$ if~and only~if there are at~least two ones among $x_1,x_2,x_3$. 
In~Boolean logic, they can be~computed as~follows:
\[s=x_1 \oplus x_2 \oplus x_3, \quad 
c = (x_1 \land x_2) \lor (x_1 \land x_3) \lor (x_2 \land x_3).\]
This leads to~a~circuit of~size seven shown~in~\Cref{fig:fa} on~the left.
The two outputs of~this circuit are computed independently:
they do~not share any gates. It~is well known that the same function
can be~computed by a~circuit of~size~$5$ shown in~\Cref{fig:fa} on~the right.
This circuit computes the carry bit as 
\[x_1x_2 \oplus x_1x_3 \oplus x_2x_3 = x_1x_2 \oplus (x_1 \oplus x_2)x_3.\]
Such an~expression for the carry bit allows the circuit to~reuse the gate computing $x_1 \oplus x_2$ for both outputs.

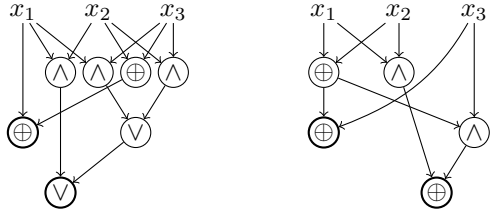
\begin{figure}[ht]
	\begin{center}
		\begin{tikzpicture}[yscale=.8]
			\begin{scope}
				\foreach \i in {1, 2, 3}
					\node[input] (\i) at (\i, 3) {$x_{\i}$};
				
				\foreach \t/\i/\x/\y/\o/\p/\q in {
					gate/4/1.5/2/{\land}/1/2,
					gate/5/2/2/{\land}/1/3,
					gate/6/2.5/2/{\oplus}/2/3,
					gate/7/3/2/{\land}/2/3,
					outgate/8/1/1/{\oplus}/1/6,
					gate/9/2.5/1/{\lor}/5/7,
					outgate/10/1.5/0/{\lor}/4/9} 
				{
					\node[\t] (\i) at (\x, \y) {$\o$};
					\draw[->] (\p) -- (\i);
					\draw[->] (\q) -- (\i);
				}
			\end{scope}
			
			\begin{scope}[xshift=40mm]
				\foreach \i in {1, 2, 3}
					\node[input] (\i) at (\i, 3) {$x_{\i}$};
				\node[gate] (4) at (1, 2) {$\oplus$}; \draw[->] (1)--(4); \draw[->] (2)--(4);
				\node[outgate] (5) at (1, 1) {$\oplus$}; \path (3) edge[->,bend left=20] (5); \draw[->] (4)--(5);
				\node[gate] (6) at (2, 2) {$\land$}; \draw[->] (1)--(6); \draw[->] (2)--(6);
				\node[gate] (7) at (3, 1) {$\land$}; \draw[->] (3)--(7); \draw[->] (4)--(7);
				\node[outgate] (8) at (2.5, 0) {$\oplus$}; \draw[->] (6)--(8); \draw[->] (7)--(8);
			\end{scope}
		\end{tikzpicture}
	\end{center}
	\caption{Two circuits, of~size~$7$ and~$5$, computing the binary representation 
		of~the sum of~three input bits.}
	\label{fig:fa}
\end{figure}

Thus, if a~given circuit contains the left circuit from \Cref{fig:fa} as a~subpart,
one can replace~it by~the right circuit from \Cref{fig:fa}, this way reducing the size
of~the original circuit. This~is exactly what our tool does. We~elaborate on~this 
in~the next section.

\subsection{Our Contribution}
In~this paper, we~present \texttt{Simplifier}, a~new open source tool for
circuit simplification that works with circuits in~both AIG and BENCH bases.
Roughly, AIG circuits use AND and NOT gates only,
whereas BENCH circuits use other popular gates like OR and XOR
(in~Section~\ref{section:formats}, we~introduce these two formats formally).
The tool runs in~just a~few seconds for circuits of~reasonable size.
Such tools are usually called low-effort to~differentiate them
from high-effort tools that exploit more expensive strategies
like simulated annealing and SAT/QBF solvers
\cite{DBLP:conf/mfcs/KulikovPS22,10.1145/3566097.3567894,DBLP:conf/sat/ReichlSS24}.

We~performed experiments on~a~wide range of~Boolean circuits, both industrial and hand-crafted.
For AIG circuits, we~first apply the state-of-the-art \texttt{ABC} framework~\cite{abc}. Our experiments show that our tool can further reduce their size by~additional 4\%, on~average.
This shows that our tool is a~reasonable complement to~\texttt{ABC}.

Though BENCH format is~popular in~practice,
we~are not aware of~any tool capable of~simplifying circuits
in~this format. For this reason, for BENCH circuits,
we~report just the size reduction achieved by~our tool.
For the datasets considered, it~is $30\%$ on~average showing
that circuits arising in~practice are usually far from being optimal.

Our simplification algorithm is based on the idea of \emph{local transformation}:
it~enumerates small \emph{subcircuits} and tries to replace every such subcircuit with
a~smaller one computing the same function. The algorithm starts with a preprocessing stage where it removes some anomalies like dangling and duplicate gates (recall \Cref{fig:simplification-examples}).
In~the main stage, the algorithm enumerates subcircuits with tree inputs of a specific type, \emph{$3$-principal subcircuits}.
Such subcircuits are inclusion maximal in the sense that they can not be extended to a larger subcircuit with three inputs (for formal definition see Section~\ref{sec:principal-subcircuits}).
Moreover, such subcircuits can be enumerated efficiently: we~prove that there is only a linear number of such subcircuits and they can be efficiently found using depth-first search.
For 3-principal subcircuits with at~most three \emph{outputs}, the algorithm looks~up a~smaller circuit in
a~precomputed database.
For all other subcircuits, the algorithm tries to identify and merge equivalent gates in order to optimize the subcircuit locally. The main stage is~repeated several times to~account for changes in~the circuit.

Note that our approach is similar to enumeration of \emph{$k$-cuts} \cite{MCB06,LD11} that roughly corresponds 
to~enumeration of~small $k$-subcircuits with only one output gate.
In~contrast, our tool enumerates subcircuits with up~to three
outputs.

To~summarize, there are two main components of our algorithm that make it
effective and efficient. First, it stores a~precomputed database of all (nearly) optimal circuits with tree inputs and at most three outputs (i.e., multi-output subcircuits).
Second, it is able to quickly scan the given circuit for subcircuits to be improved,
due to the notion of 3-principal subcircuits discussed above.

\subsection{Related work}
Apart from the \texttt{ABC} framework we~consider in~this paper,
we~are unaware of~any other low-effort circuit minimization tools.
At~the same time, 
there exist also 
high-effort simplification tools \cite{DBLP:conf/mfcs/KulikovPS22,10.1145/3566097.3567894,DBLP:conf/sat/ReichlSS24}
that use more time-consuming strategies like simulated annealing and SAT/QBF solvers.


\section{Circuit Representation Formats}\label{section:formats}
By~$B_{n,m}$ we~denote the set of~all Boolean functions with $n$~inputs
and $m$~outputs. $B_n$ is~the set $B_{n,1}$ of~all Boolean predicates.

The two widely adopted formats for representing Boolean circuits are
AIG~\cite{biere2007aiger} and BENCH~\cite{benchformat}.
\begin{description}
	\item[AIG] stands for and-inverter graph. In~this format, a~circuit uses binary ANDs
	and unary NOTs (inverters) as~its gates. In~other words, it~uses a~basis $\{\land, \neg\}$. An~AIG circuit
	is~essentially a~graph of~indegree~two with binary labels on~the edges that correspond to negations. Thus, every node of~this graph computes the binary conjunction.
	This makes~it particularly convenient for storing and drawing such circuits.
	The size of an~AIG circuit is~the number of~ANDs, that~is, the number
	of~nodes in~the graph.

	\item[BENCH] circuits use the following set of gates: unary NOT, binary AND, OR, XOR, and their negations. This is a~richer basis and many functions have more compact representations in~the BENCH format than in~the AIG format (recall Figure~\ref{fig:parity-5}).
\end{description}


\section{Main Algorithm}

\subsection{Database of (Nearly) Optimal Circuits}
Below, we~describe the way we~computed (nearly) optimal circuits for
all Boolean functions with three inputs and at~most three outputs.
The case when the number of~outputs is~smaller than three can
be~reduced easily to~the case when there are exactly three outputs
(by~adding one or~two dummy outputs), so~we assume that the number of~outputs
is~equal to~three. The number of~such functions~is $\binom{2^{2^3}}{3}=2\,763\,520.$
In~order to~reduce the search space,
we~use an~idea of~partitioning the set $B_{3,3}$
of~all functions $\{0,1\}^3\to \{0,1\}^3$ into classes having the same circuit complexity due to~Knuth~\cite[Section~7.1.2]{Knuth:2008:ACP:1377542}.
For example, two functions have the same circuit size if~they result from each other by~permuting their inputs or~outputs.
For the AIG basis,
the circuit size does not change if~one negates some of~the inputs and outputs.
This allows one to~narrow the search space to~functions that output $(0,0,0)$ on~the input $(0,0,0)$.
Knuth calls such functions \emph{normal}.
Using a~similar partition into classes, Knuth computed the exact circuit complexity over
the basis~$B_2$ of~all functions from $B_{4}$~and~$B_5$.

In~the BENCH basis, each function from $B_{3,3}$ has a~relatively small circuit, and it~is possible to~find a~provably
optimal circuit using a~SAT solver.
We~used a~tool~\cite{DBLP:conf/mfcs/KulikovPS22} for this task, where we searched for the smallest circuit size $n$
such that the SAT solver could both find a~circuit of~size $n$ and prove that no~circuit of~size $n - 1$ exists.

In~the AIG basis, circuits for functions from $B_{3,3}$ are typically larger, making it currently infeasible
to compute provably optimal circuits for all of them using a SAT solver.
While we applied the same search for the smallest circuit size $n$, the SAT solver was often unable to prove
the absence of smaller circuits or would hang indefinitely.
Consequently, for several classes of functions, we have efficient circuits but lack formal proof that they are of minimal size.

We~provide detailed statistics in~\Cref{table:stat}.
As~the table reveals, each function from $B_{3,3}$
can be~computed by a~BENCH circuit of~size at~most~$8$
and an~AIG circuit of~size at~most~$11$. \Cref{fig:hard}
shows two circuits computing functions from $B_{3,3}$
and having the maximum circuit complexity in~the
BENCH and AIG bases.

\begin{table}[!ht]
	\begin{center}
		\begin{tabular}{rrrrr}
			\toprule
            & \multicolumn{2}{c}{BENCH} & \multicolumn{2}{c}{AIG}\\
            \cmidrule(lr){2-3} \cmidrule(lr){4-5}
			\# gates & \# classes & \# functions & \# classes &  \# functions \\
			\midrule
			2 & 45 & 396 & 0 & 0 \\
			3 & 659 & 12,480 & 51 & 11,840 \\
			4 & 4,541 & 152,504 & 232 & 72,264 \\
			5 & 18,056 & 761,656 & 726 & 223,640 \\
			6 & 29,571 & 1,349,492 & 1,540 & 498,720 \\
			7 & 10,409 & 481,824 & 2,318 & 762,128 \\
			8 & 119 & 5,168 & 2,163 & 727,264 \\
			9 & 0 & 0 & 16 & 4,672 \\
			$\le 9$ & 0 & 0 & 1610 & 392208 \\
			$\le 10$ & 0 & 0 & 224 & 69664 \\
			$\le 11$ & 0 & 0 & 6 & 1120 \\
			\bottomrule
		\end{tabular}
	\end{center}

	\caption{Distribution of~classes and functions by~circuit size {(the number of gates in a~circuit)}, for all functions with three inputs and three distinct outputs, for BENCH and AIG bases.}
	\label{table:stat}
\end{table}

\begin{figure}[!ht]
	\begin{center}
		\begin{tikzpicture}
		\begin{scope}[yscale=.7]
		\node[input] (x1) at (0, 3) {$x_1$};
		\node[input] (x2) at (1, 3) {$x_2$};
		\node[input] (x3) at (2, 3) {$x_3$};

		\node[gate] (g3) at (0.5, 2) {$\oplus$};
		\node[gate] (g4) at (1.5, 2) {$=$};

		\node[gate] (g5) at (0, 1) {$\overline\lor$};
		\node[gate] (g8) at (1, 1) {$\lor$};

		\node[gate] (g6) at (0.5, 0) {$\overline\lor$};
		\node[outgate] (g9) at (1.5, 0) {$=$};

		\node[outgate] (g7) at (2, -1) {$=$};

		\node[outgate] (g10) at (0, -1.5) {$\lor$};
		\foreach \s/\t in {x1/g3, x2/g3, x1/g4, x3/g4, x1/g5, g4/g5, g3/g8, g4/g8, x3/g9, g8/g9, g5/g6, g3/g6, g6/g7, x3/g7, g5/g10, g7/g10}
		    \draw[wire] (\s) -- (\t);

		\end{scope}

		\begin{scope}[shift={(4,-1.5)},yscale=0.7]

		\node[input] (x1) at (0, 6) {$x_1$};
		\node[input] (x2) at (2, 6) {$x_2$};

		\node[gate] (g7) at (0, 5) {$\land$};
		\node[gate] (g10) at (1, 5) {$\land$};
		\node[input] (x3) at (2, 5) {$x_3$};

		\node[gate] (g8) at (0, 4) {$\land$};
		\node[gate] (g5) at (1, 4) {$\land$};
		\node[gate] (g6) at (2, 4) {$\land$};
		\node[gate] (g11) at (3, 4) {$\land$};

		\node[outgate] (g12) at (0.5, 3) {$\land$};
		\node[gate] (g4) at (1.5, 3) {$\land$};

		\node[gate] (g13) at (0, 2) {$\land$};

		\node[outgate] (g3) at (0, 1) {$\land$};
		\node[outgate] (g9) at (2, 1) {$\land$};

		\foreach \s/\t in {x1/g10, x2/g10, x3/g8, x3/g6, x3/g11, g10/g6}
		    \draw[wire] (\s) -- (\t);

		\foreach \s/\t in {x1/g7, x2/g7, g7/g8, g10/g5, x3/g5, x2/g11, g8/g12, g5/g12, g5/g4, g6/g4, g4/g13, g13/g3, g13/g9, g11/g9}
		    \draw[notwire] (\s) -- (\t);

		\draw[notwire] (g7) to [bend right=20] (g13);
		\draw[notwire] (g7) to [bend right=30] (g3);
		\end{scope}
		\end{tikzpicture}
	\end{center}

	\caption{Two circuits computing two different functions from $B_{3,3}$. The left circuit is~over the base BENCH, 
		has size~$8$ and is~\emph{provably} optimal. The right circuit
		is~over the base AIG, has size~$11$, and 
		is~\emph{presumably} optimal:
		we~are not aware of a~smaller AIG circuit for the same function, but proving this through checking all circuits of~size~$10$
		is~infeasible in~practice.}
	\label{fig:hard}
\end{figure}
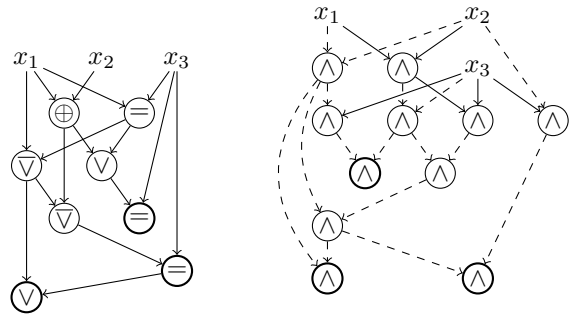

\subsection{Local Transformation}
Our simplification algorithm is based on the idea of \emph{local transformation}---it enumerates small \emph{subcircuits} and tries to replace every such subcircuit with a smaller one computing the same function. In the following, we explain this idea in more detail.

Given a~subset~$G$ of~the gates of a~circuit~$C$,
we~say that \emph{$G$~induces a~subcircuit of~$C$},
if the following property
holds for each gate $v \in G$: if~$v$ is not an~input gate of~$C$, then
either both predecessors of~$v$ belong to~$G$
or~both of~them do~not belong to~$G$.
The \emph{input gates} of the~subcircuit induced by~$G$ are the gates that have no~predecessors in~$G$.
The \emph{output gates} are those gates in~$G$ that
have~at least one successor outside of $G$ and those that are outputs of $C$.
A~\emph{$k$-subcircuit} is a~subcircuit with $k$~input gates. See Fig.~\ref{fig:induced-subcircuits} for an example.
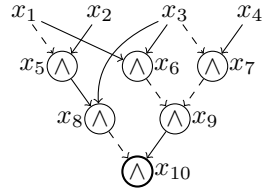
\begin{figure}[ht]
    \centering
\begin{tikzpicture}
\begin{scope}[yscale=.7]
\node[input] (x1) at (0, 6) {$x_1$};
\node[input] (x2) at (1, 6) {$x_2$};
\node[input] (x3) at (2, 6) {$x_3$};
\node[input] (x4) at (3, 6) {$x_4$};

\node[gate, label={[xshift=0.12cm]left:$x_5$}] (g5) at (0.5, 5) {$\land$};
\node[gate, label={[xshift=-0.11cm]right:$x_6$}] (g6) at (1.5, 5) {$\land$};
\node[gate, label={[xshift=-0.11cm]right:$x_7$}] (g7) at (2.5, 5) {$\land$};

\node[gate, label={[xshift=0.12cm]left:$x_8$}] (g8) at (1, 4) {$\land$};
\node[gate, label={[xshift=-0.11cm]right:$x_9$}] (g9) at (2, 4) {$\land$};

\node[outgate, label={[xshift=-0.11cm]right:$x_{10}$}] (g10) at (1.5, 3) {$\land$};

\foreach \s/\t in {x2/g5, x1/g6, x3/g6, x4/g7, g5/g8, g9/g10}
    \draw[wire] (\s) -- (\t);

\foreach \s/\t in {x1/g5, x3/g7, g6/g9, g7/g9, g8/g10}
    \draw[notwire] (\s) -- (\t);
\draw[wire] (x3) to [bend right=30] (g8);
\end{scope}
\end{tikzpicture}
\caption{$G=\{x_1, x_2, x_3, x_5, x_6, x_8\}$
    induces a~subcircuit with inputs $\{x_1,x_2, x_3\}$ and outputs $\{x_3, x_6, x_8\}$,
    whereas $G=\{x_6, x_7, x_8, x_9, x_{10}\}$ induces a~subcircuit with inputs $\{x_6, x_7, x_8\}$ and output $\{x_{10}\}$.}\label{fig:induced-subcircuits}
\end{figure}

The reason for considering subcircuits is~simple.
A~subcircuit~$C'$ of~$C$ is~itself a~circuit computing some
Boolean function $h$.
If $h$~can be computed by~a~smaller circuit $C''$,
then we can replace~an instance of $C'$ in~$C$.
Our algorithm enumerates all $3$-principal subcircuits (see \Cref{sec:principal-subcircuits}) of the input circuit.
One can show that the number of $3$-principal subcircuits is linear (see \Cref{thm:3-principal}) in the size (number of gates) of a circuit.
For every $3$-principal subcircuit with at most three outputs, the algorithm computes its truth table and looks up a smaller circuit computing the same function in a precomputed circuit database.
For subcircuits that have more than three outputs, the algorithm tries to optimize the subcircuit locally by identifying and merging equivalent gates. If after that the number of output gates is at most three, the algorithm looks up the database.

At~first, the algorithm finds all the subcircuits that can be optimized and then it tries to optimize them one by one. As~the subcircuits might intersect, optimizing one of the subcircuits might change some other subcircuit. The algorithm keeps track of it and skips the subcircuits that were modified due to optimization of other subcircuits.
If after one cycle of this algorithm, the circuit changes then there might be some new 3-principal subcircuits that can be optimized. For this reason, we run this algorithm several times in a row (be default we run it 5 times).
Our experiments show that an~impact of~every next iteration is usually significantly smaller {(see \Cref{subsection:experiments}, and \Cref{table:iterationStat})}, however this helps to improve the parts of the circuit that correspond to the subcircuits skipped in the previous iteration.

\subsection{Principal Subcircuits}\label{sec:principal-subcircuits}
Throughout this section, let $C$ be a~fixed circuit.
For a~subset~$X$ of~the gates of~$C$,
define $g(X)$ to be the set of all gates such that
\begin{itemize}
    \item every $v\in g(X)$ either belongs to $X$ or has a predecessor in $X$,
    \item given any assignment to gates in $X$, one can uniquely determine a value for every gate in $g(X)$.
\end{itemize}

We say that a set of gates $X$ is a \emph{proper subcircuit generator}
if there is a gate $v\in g(X)$ such that every $Y \subsetneq X$, $v\not\in g(Y)$.
A proper subcircuit generator $X$ induces a subcircuit on gates $g(X)$ with inputs $X$.
We will identify $g(X)$ with the corresponding induced subcircuit.

We are going to consider only proper subcircuit generators of a small size. For a gate $v$, let $S_k(v)$ be the set of proper subcircuit generators of size $k$ such that for every $X\in S_k(v)$, $v\in g(X)$ and for any $Y\subsetneq X$, $v\not\in g(Y)$.

In the case of $k=2$, one can show that $S_k(v)$ always contains the topologically highest proper subcircuit generator.
\begin{theorem}\label{thm:2-principal}
    For every gate $v$ with $S_2(v)\neq \emptyset$, there is $X \in S_2(v)$ such that for all $Y\in S_2(v)$, $g(Y)\subseteq g(X)$.
\end{theorem}
For the topologically highest proper subcircuit generator $X$ provided by \Cref{thm:2-principal}, we say that $g(X)$ is a \emph{$2$-principal subcircuit}.

Similarly, we want to define $3$-principal subcircuit. For the case of $k=3$, it is not true that there always is a topologically highest proper subcircuit generator in $S_3(v)$. Luckily, it is possible to show that there are at most two ``maximal'' proper subcircuit generators.

\begin{theorem}\label{thm:3-principal}
    For every gate $v$ with $S_3(v)\neq \emptyset$, there are $X,Y \in S_3(v)$ such that for all $Z\in S_3(v)$, $g(Z)\subseteq g(X) \lor g(Z)\subseteq g(Y)$.
\end{theorem}
For $X$ and $Y$ provided by \Cref{thm:3-principal}, we say that the subcircuits induced by $g(X)$ and $g(Y)$ are \emph{$3$-principal}. See Fig.~\ref{fig:principal} for examples.

Immediately from \Cref{thm:3-principal} we can conclude that every circuit with $n$ gates contains at most $2n$ $3$-principal subcircuits. These subcircuits can be efficiently enumerated by an algorithm based on depth-first search: the set $S_3$ for a gate can be computed from the $S_3$ sets of its parents (the case of a gate with only one parent is trivial).

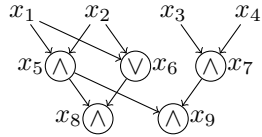
\begin{figure}[!ht]
    \centering
    \begin{tikzpicture}
        \begin{scope}[yscale=.7]
            \node[input] (x1) at (0, 6) {$x_1$};
            \node[input] (x2) at (1, 6) {$x_2$};
            \node[input] (x3) at (2, 6) {$x_3$};
            \node[input] (x4) at (3, 6) {$x_4$};

            \node[gate, label={[xshift=0.12cm]left:$x_5$}] (g5) at (0.5, 5) {$\land$};
            \node[gate, label={[xshift=-0.11cm]right:$x_6$}] (g6) at (1.5, 5) {$\lor$};
            \node[gate, label={[xshift=-0.11cm]right:$x_7$}] (g7) at (2.5, 5) {$\land$};

            \node[gate, label={[xshift=0.12cm]left:$x_8$}] (g8) at (1, 4) {$\land$};
            \node[gate, label={[xshift=-0.11cm]right:$x_9$}] (g9) at (2, 4) {$\land$};

            \foreach \s/\t in {x2/g5, x1/g6, x2/g6, x4/g7, g5/g8}
            \draw[wire] (\s) -- (\t);

            \foreach \s/\t in {x1/g5, x3/g7, g5/g9, g7/g9, g6/g8}
                \draw[wire] (\s) -- (\t);

        \end{scope}
    \end{tikzpicture}
    \caption{For the gate $x_8$, $S_2(x_8) = \{\{x_1,x_2\},\{x_5,x_6\}\}$. The subcircuit induced by $g(\{x_1,x_2\}) =\{x_1,x_2,x_5,x_6,x_8\}$ is $2$-principal, while
        the subcircuit induced by $g(\{x_5,x_6\}) = \{x_5,x_6,x_8\}$ is not.
        For the gate $x_9$, $S_3(x_9) = \{\{x_1,x_2,x_7\},\{x_5,x_3,x_4\}\}$. The subcircuits induced by
        $\{x_1,x_2,x_5,x_6,x_8,x_7,x_9\}$ and $\{x_3,x_4,x_5,x_7,x_9\}$ are $3$-principal.}\label{fig:principal}
\end{figure}

In the rest of this section we explain how to prove Theorem~\ref{thm:3-principal}. We start by proving Theorem~\ref{thm:2-principal}, and then we give a sketch of the proof for Theorem~\ref{thm:3-principal}.

The proof of Theorem~\ref{thm:2-principal} depends on the following two lemmas.
The first lemma follows from the definition of $g(X)$.
\begin{lemma}\label{lm:all-paths-go-through-X}
	Let $X$ be a set of gates. For any gate $v$, $v\in g(X)$ if and only if every ascending path from $v$ to the inputs goes through some gate from $X$.
\end{lemma}
\begin{proof}
If $v\in g(X)$ then either $v\in X$ or all parents of $v$ are in $g(X)$. Hence all the ascending paths necessarily go through the gates in $X$. If all the ascending paths from $v$ go through the gates in $X$, then given an assignment to $X$ we can compute a value for $v$, and hence $v\in g(X)$.
\end{proof}

The second lemma shows that all the gates of a proper subcircuit generator are essential.
\begin{lemma}\label{lm:paths-for-elements-of-X}
	Let $X$ be a proper subcircuit generator and $v\in g(X)$ be such that $\forall Y\subsetneq X$, $v\not\in g(Y)$. Then for every $x\in X$ there is an ascending path $\pi$ from $v$ to some input such that $\pi \cap X = \{x\}$.
\end{lemma}
\begin{proof}
By Lemma~\ref{lm:all-paths-go-through-X}, all the ascending paths from $v$ go through $X$. The condition $\forall Y\subsetneq X$, $v\not\in g(Y)$ guarantees that for every $x\in X$ there is an ascending path from $v$ going through $x$. Finally, assume that for some $x \in X$ all the ascending paths from $v$ go through some other gate in $X$. Let $Y\subset X\setminus\{x\}$ be the set of all gates that belong to such ascending paths from $v$ going through $x$. There are two cases: either $V$ contains a cut separating $v$ and $x$, or $V$ contains a cut separating $x$ and the inputs. In both cases, some of the gates in $X$ are redundant which contradicts the definition of a proper subcircuit generator.
\end{proof}

\begin{proof}[Theorem~\ref{thm:2-principal}]
Assume for the sake of contradiction that there are two distinct topologically highest 2-gate sets $X$ and $Y$ such that $g(X)\not\subset g(Y)$ and $g(Y)\not\subset g(X)$.
Let's start with the case when $X$ and $Y$ share some gate, i.e. $X=\{a,b\}$ and $Y=\{a,c\}$.
Consider an ascending path from $v$ that goes trough gate~$b$ and avoids gate~$a$. Such a~path is guaranteed to exist by Lemma~\ref{lm:paths-for-elements-of-X}. By Lemma~\ref{lm:all-paths-go-through-X} this path has to go through gate $c$ (otherwise it avoids all the gates in $Y$). If $c$ precedes $b$ on the path then $c\in g(X)$ and hence $g(Y)\subset g(X)$ which contradicts the assumption. Similarly, if $b$ precedes $c$ on the path then we also get a contradiction $g(X)\subset g(Y)$.

The analysis for the case of disjoint $X$ and $Y$ is similar. Let $X=\{a,b\}$ and $Y=\{c,d\}$. We need to consider all mixed couples $\{a,c\}, \{a,d\}, \{b,c\}, \{b,d\}$ and by careful case analysis show that one of them dominates both $X$ and $Y$.
\end{proof}

Now we describe a proof sketch for Theorem~\ref{thm:3-principal}.
\begin{proof}[Theorem~\ref{thm:3-principal}, sketch]
    For a set of gates $X$ and a gate $v\in g(X)$, we denote by $dd(v,X)$ the \emph{dependency degree} of $v$ on $X$ that is the minimum size of $Y\subset X$ such that $v\in g(Y)$. Let $v$, $X$ and $Y$ be defined as in  Theorem~\ref{thm:3-principal}, and also $X\neq Y$. We are going to annotate some gates with pairs of numbers that are the dependency degrees on $X$ and $Y$. Lets start with $v$, it is annotated with $(dd(v,X), dd(v,Y)) = (3,3)$. Then we proceed to the parents of $v$ and continue until we find a gate such that it is annotated with $(3,3)$ but both its parents annotated with something else. Let $w$ be the first such gate on some ascending path from $v$, and let $a_1$ and $a_2$ be the annotations of $w$'s parents. One can prove the following lemma.
    \begin{lemma}\label{lm:3-principal-lemma}
        $\{a_1,a_2\} \in
        \bigl\{
           \{(2,1),(1,2)\}, \{(2,1),(1,3)\}, \{(1,2),(3,1)\}
        \bigr\}$.
    \end{lemma}
    The proof of this lemma is a careful application of Lemmas~\ref{lm:all-paths-go-through-X} and~\ref{lm:paths-for-elements-of-X}, and Theorem~\ref{thm:2-principal} to the case analysis.

    Assume for the sake of contradiction that there are three distinct proper subcircuit generators that do not dominate each other $X$, $Y$ and $Z$. Now we are going to annotate gates with triples of numbers that are the dependency degrees on $X$, $Y$ and $Z$. Gate $v$ is annotated with $(dd(v,X),dd(v,Y),dd(v,Z)) = (3,3,3)$. Again we proceed to the parents of $v$ and continue until we find a gate such that it is annotated with $(3,3,3)$ but both its parents are annotated with something else. Now our goal is to analyze all the cases and show that every case is either impossible, e.g., due to Lemma~\ref{lm:3-principal-lemma}, or in this case we can construct a new triple from the pieces of $X$, $Y$ and $Z$ that dominates some of them.
\end{proof}


\section{Experimental Evaluation}
In~this section, we~report the results of~experiments.
The used datasets, the source code, and the scripts for statistics collection
are provided in~the~artifact attached to~this paper.
The code and benchmarks are also available in the standalone repository \url{https://github.com/SPbSAT/simplifier}.

\subsection{Datasets}\label{section:datasets}

For the experiments, we~used the following collections of~circuits.
\begin{description}
	\item[Combinatorial principles.] We~implemented circuit generators for a~number of~combinatorial objects. Below,
    we~describe the corresponding circuits. Besides the standard
    conjunction, disjunction, and negation gates, we~use the following two building blocks.
    \begin{itemize}
        \item $\SUM(x_1, \dotsc, x_n) \colon \{0,1\}^n \to \{0,1\}^{\lceil \log_2(n+1) \rceil}$ computes the binary representation of~the sum (over integers) of~$n$~input bits. It~can
        be~computed by~circuits of~size $4.5n+o(n)$ over~$B_2$~\cite{DBLP:journals/ipl/DemenkovKKY10} and circuits of~size~$7n+o(n)$
        over~$\{\land, \lor, \neg\}$~\cite{DBLP:journals/siamcomp/Zwick91}.
        \item $\AtLeast_k(x_1, \dotsc, x_n) \colon \{0,1\} \to \{0,1\}=[x_1+\dotsc+x_n \ge k]$, i.e., it~outputs~1 if~and only if~the sum of the~$n$~input bits is~at~least~$k$. It~can be~computed in~$o(n)$ additional gates after computing $\SUM(x_1, \dotsc, x_n)$. When $k=1$, the function simplifies to~the disjunction of~the input bits.
        The function $\AtMost_k$ is~defined similarly.
    \end{itemize}

    We~use the following combinatorial principles and problems.
	\begin{description}
		\item[Pigeonhole principle:] given integers $n$,~$m$, and~$k$,
        a~circuit checks
		whether it is possible
        to~distribute $n$~pigeons into $m$~holes so~that
        no~hole contains more
        than $k$~pigeons.
		The circuit has $n m$ inputs, where an~input variable $x_{ij}$ indicates whether the $i$-th pigeon is placed in the $j$-th hole. The output is defined by the following function:
		\[\bigwedge_{i = 1}^n \AtLeast_1(x_{i1}, x_{i2}, \dotsc, x_{im}) \land \bigwedge_{j = 1}^m \AtMost_k(x_{1j}, x_{2j}, \dotsc, x_{nj}).\]
        Clearly, the circuit is~satisfiable if~and only~if $n \ge mk$.

        \item[Even colouring principle:] given an~Eulerian undirected graph $G = (V, E)$, a~circuit checks whether
		the edges of~$G$ can be labeled by~$0$ and~$1$ in such a~way that
		half of the edges incident to any vertex have label~$1$.
        We~introduce a~variable $x_{uv}$ for every edge $\{u,v\} \in E$.
        Then, the output is~defined by the following function:
		\[\bigwedge_{u = 1}^n \left[\SUM(\{ x_{uv} \colon \{u,v\} \in E\}) = \frac{\operatorname{degree}(u)}{2}\right].\]
        It is known that an~Eulerian graph has an~even colouring if~and only~if
        it~has an~even number of~edges~\cite{ecp}.
		Most of~our benchmarks are generated using
		random $d$-regular graphs for even~$d$.

		\item[Clique problem:] given an~undirected graph $G = (V, E)$ with
        $n$~nodes and an~integer parameter~$k$, a~circuit checks whether there exists a~clique
        of~size~$k$ in~$G$.
		For each $v \in V$, we~introduce a~variable~$x_v$ indicating
        whether $v$~belongs to a~clique.
		The output is defined by the following function:
		\[\bigwedge_{e =\{u,v\} \not \in E} (\overline{x_u} \lor \overline{x_v}) \land \AtMost_{n-k}(\{\overline{x_v} \colon v \in V\}).\]
		Our benchmarks are generated using graphs with large cliques~\cite{cliquebase}.

		\item[Factorization:] given an~integer~$k$, the circuit checks whether there exists a decomposition of~$k$ into two nontrivial factors. We use a~circuit computing the product of two $\lceil \log k \rceil$-bit integers, with the added constraint that the first multiplier is~neither~$1$ nor~$k$.
	\end{description}
	\item[Miter circuits.] Given two circuits $C_1, C_2 \colon \{0,1\}^n \to \{0,1\}^m$, their miter is a~circuit~$C$ combined out~of $C_1$~and~$C_2$ which is~unsatisfiable if~and only~if $C_1$~and~$C_2$ compute the same function. Formally, for an~input $x \in \{0,1\}^n$, let $C_1(x)=(y_1, \dotsc, y_m)$ and $C_2(x)=(z_1, \dotsc, z_m)$. Then,
    \(C(x)=\bigvee_{i=1}^{m}(y_i \oplus z_i).\)
    We~use miters for various pairs of~circuits described below.
	\begin{description}
		\item[Summation.] Circuits of size $4.5n+o(n)$ and $5n$ computing the binary representation of~the sum of $n$~input bits~\cite{DBLP:journals/ipl/DemenkovKKY10}.
		\item[Threshold.] Circuits of~size $2n+o(n)$ and $3n-O(1)$ computing $[x_1+\dotsb+x_n \ge 2]$.
		\item[Multiplication.] Circuits computing the product of~two $n$-bit integers of~size~$n^2$ (high school multiplication method) and of~size $n^{1.66}$ (Karatsuba's algorithm). Some of~the circuits are produced using Transalg tool~\cite{lmcs:6177}.
		\item[Sorting.] Circuits sorting fixed-size integers. The circuits are
		produced from various sorting algorithms using Transalg~\cite{lmcs:6177}.
	\end{description}

	\item[Public datasets.] We~also used many publicly available collections of~circuits:
    combinational benchmarks from ISCAS85~\cite{iscas85},
    EPFL combinational benchmark suite~\cite{epflBench},
    benchmarks from HWMCC20~\cite{hwmcc20},
    ITC99 benchmarks~\cite{ITC99},
    Bristol Fashion MPC circuits~\cite{BristolFashion}.
\end{description}

\subsection{Experimental Results}\label{subsection:experiments}
Below, we~report the results of~experiments for the two bases, AIG and BENCH.

\subsubsection{AIG Circuits}
	We compare our tool, \texttt{Simplifier}, with the commonly used \texttt{resyn2} script~\cite{10247961} from the \texttt{ABC} framework.
	Following~\cite{8342109,LD11,10.1145/3566097.3567894}, we~perform several consecutive runs of~\texttt{resyn2} on AIG circuits and compare them to collaborative runs of \texttt{resyn2} and \texttt{Simplifier}: this ensures that when we~apply our tool
    to~a~circuit, this circuit is~already simplified by~the state-of-the-art \texttt{ABC} framework.
    For experimental evaluation, we~used all circuits (from the corresponding class) having no~more than $50,000$ gates to ensure that the calculations could be completed in a reasonable amount of time. 
    
    \Cref{table:stat1AIG} shows the results of~experimental evaluation using the following metrics.
	Let $s_0$~denote the size of the original circuit, $s_1$~be the size after the first sequence, and $s_2$~be the size after the second sequence that includes a~call to~\texttt{Simplifier} (where the size refers to the number of AND-gates).
	Let also $t_1, t_2$~be the time required for the first and the second sequences respectively.
    We~report the \emph{average} and \emph{median} values for three metrics to~evaluate the performance of \texttt{Simplifier}:

    \begin{description}
	\item[Absolute improvement]  
	quantifies the additional size reduction achieved by using \texttt{Simplifier} after applying \texttt{resyn2}, calculated as $(s_1 - s_2) / s_1$. It highlights how much smaller a circuit becomes when combining \texttt{Simplifier} with \texttt{resyn2} rather than relying solely on \texttt{resyn2}. This provides a direct measure of the impact of \texttt{Simplifier}.
	
        \item[Relative improvement]
	compares the reduction achieved by \texttt{Simplifier} relative to the reduction already achieved by \texttt{resyn2}, calculated as $(s_1 - s_2) / (s_0 - s_1)$. Unlike absolute improvement, which depends on the size of the original circuit, this metric normalizes the comparison to focus on the relative effectiveness of \texttt{Simplifier}. It is particularly useful for understanding how impactful \texttt{Simplifier} is compared to \texttt{resyn2}, regardless of circuit size.
	
    \item [Time overhead]
	    evaluates the additional computational time required for the second sequence (including \texttt{Simplifier}) relative to the first sequence (using only \texttt{resyn2}), calculated as $(t_2 - t_1) / t_1$. It helps assess the efficiency of incorporating \texttt{Simplifier} in terms of time cost.
    \end{description}

    The results show that, in~many cases, our tool can further simplify a~given circuit,
    making the tool a~valuable complement to~\texttt{ABC}.

\begin{table}[ht]
    \begin{center}
        \setlength\tabcolsep{.75\tabcolsep}
        \begin{tabular}{p{42mm}rrrrrrr}
			\toprule
			& & \multicolumn{3}{c}{average (\%)} & \multicolumn{3}{c}{median (\%)} \\
			\cmidrule(lr){3-5} \cmidrule(lr){6-8}
			Class & Count &
			\multicolumn{1}{c}{relative} & \multicolumn{1}{c}{absolute} & \multicolumn{1}{c}{time} &
			\multicolumn{1}{c}{relative} & \multicolumn{1}{c}{absolute} & \multicolumn{1}{c}{time} \\
			\midrule
			\multicolumn{8}{c}{comparing $\texttt{resyn2} \times 6$ against $\texttt{resyn2} \times 6 + \texttt{Simplifier}$} \\
			\midrule
			Pigeonhole principle & 60 & $0.18$ & $0.03$ & $12.43$ & $0.00$ & $0.00$ & $10.20$ \\
			Even colouring principle & 67 & $0.78$ & $0.22$ & $57.07$ & $0.00$ & $0.00$ & $43.62$ \\
			Clique search problem & 52 & $0.72$ & $1.20$ & $17.24$ & $0.62$ & $1.30$ & $15.43$ \\
			Factorization &129  & $0.11$ & $0.09$ & $33.15$ & $0.11$ & $0.08$ & $35.85$ \\
			Miter, summation & 7 & $31.16$ & $8.22$ & $58.54$ & $31.22$ & $8.22$ & $62.76$ \\
			Miter, threshold & 16 & $5.60$ & $0.03$ & $54.61$ & $1.09$ & $0.01$ & $50.94$ \\
			Miter, multiplication & 43 & $5.79$ & $3.54$ & $18.90$ & $8.01$ & $4.72$ & $17.98$ \\
			Miter, sorting & 21  & $0.31$ & $0.11$ & $29.30$ & $0.19$ & $0.09$ & $27.85$ \\
			EPFL & 28 & $490.18$ & $1.29$ & $45.31$ & $0.16$ & $0.00$ & $48.44$ \\
			HWMCC’20 & 17 & $0.40$ & $0.19$ & $48.78$ & $0.00$ & $0.00$ & $50.70$ \\
			ITC’99 & 25 & $7.42$ & $0.98$ & $51.27$ & $4.99$ & $0.79$ & $31.08$ \\
			Bristol Fashion MPC Circuits & 9 & $1.88$ & $1.50$ & $35.15$ & $0.11$ & $0.07$ & $26.47$ \\
			\midrule
			\multicolumn{8}{c}{comparing $\texttt{resyn2} \times 3$ against $\texttt{resyn2} \times 2 + \texttt{Simplifier}$} \\
			\midrule
			Pigeonhole principle & 60 & $-0.19$ & $-0.85$ & $-8.47$ & $0.00$ & $0.00$ & $-11.84$ \\
			Even colouring principle & 67 & $-0.06$ & $-1.63$ & $56.62$ & $-0.02$ & $-0.08$ & $46.20$ \\
			Clique search problem & 52 & $0.31$ & $0.24$ & $-0.24$ & $0.15$ & $0.11$ & $-1.40$ \\
			Factorization & 129 & $0.07$ & $0.09$ & $24.62$ & $0.06$ & $0.08$ & $18.37$ \\
			Miter, summation & 7 & $8.22$ & $31.16$ & $61.33$ & $8.22$ & $31.22$ & $59.72$ \\
			Miter, threshold & 16 & $-0.18$ & $-30.62$ & $42.87$ & $-0.02$ & $-8.71$ & $35.75$ \\
			Miter, multiplication & 43 & $3.40$ & $5.61$ & $-0.23$ & $4.54$ & $7.63$ & $0.74$ \\
			Miter, sorting & 21 & $-0.59$ & $-1.47$ & $15.18$ & $-0.34$ & $-0.96$ & $10.77$ \\
			EPFL & 28 & $1.07$ & $489.47$ & $41.54$ & $0.00$ & $0.00$ & $33.53$ \\
			HWMCC’20 & 17 & $-0.04$ & $-0.01$ & $46.27$ & $0.00$ & $0.00$ & $29.21$ \\
			ITC’99 & 25 & $0.69$ & $5.76$ & $59.29$ & $0.41$ & $3.73$ & $26.32$ \\
			Bristol Fashion MPC Circuits & 9 & $1.41$ & $1.80$ & $30.52$ & $0.00$ & $0.00$ & $19.23$ \\
			\bottomrule
		\end{tabular}
	\end{center}
	\caption{Comparison of several runs of~\texttt{resyn2} against several runs of~\texttt{resyn2} followed by~a~run of~\texttt{Simplifier}. All values are percentages.}
	\label{table:stat1AIG}
\end{table}

    Our experiments suggest that, in~most cases, one does not need to~run \texttt{Simplifier}
    for more than five iterations: most of~the size reduction happens during the first iteration and the reduction decays during the next four (or~less) iterations.
    Table~\ref{table:iterationStat} shows the details for twelve circuits
    selected out of twelve classes considered above.

    \begin{table}[!ht]
        \begin{center}
            \setlength\tabcolsep{.9\tabcolsep}
            \begin{tabular}{p{35mm}rrrrrrrrrrr}
                \toprule
                & \multicolumn{3}{c}{size (gates)} & \multicolumn{2}{c}{time (sec.)} & \multicolumn{5}{c}{iterations}\\
                \cmidrule(lr){2-4} \cmidrule(lr){5-6} \cmidrule(lr){7-11}
                name & original & \texttt{r} & \texttt{r+s} & \texttt{r} & \texttt{s} & 1 & 2 & 3 & 4 & 5 \\
                \midrule
                19\_20\_1.bench & 2792 & 2195 & 2176 & 1.04 & 0.33 & 19 & 0 & - & - & - \\
                900\_4.bench & 14399 & 12463 & 12459 & 1.22 & 2.87 & 4 & 0 & - & - & - \\
                hamming8-4\_sat.bench & 25320 & 7660 & 7519 & 2.99 & 1.84 & 75 & 62 & 3 & 0 & - \\
                2356808270112211.bench & 37318 & 21064 & 21039 & 6.13 & 6.96 & 13 & 6 & 0 & - & - \\
                1500.bench & 17147 & 13760 & 12456 & 5.08 & 6.21 & 497 & 187 & 186 & 0 & - \\
                thr2\_6000.bench & 30135 & 30125 & 30124 & 8.30 & 11.32 & 1 & 0 & - & - & - \\
                trVlog\_20.bench & 9849 & 6207 & 5865 & 1.82 & 1.33 & 156 & 158 & 26 & 0 & - \\
                PvS\_5\_4-aigmiter.bench & 1628 & 1217 & 1203 & 0.28 & 0.18 & 10 & 2 & 0 & - & - \\
                voter\_size\_2023.bench & 37345 & 21645 & 21440 & 3.97 & 18.42 & 109 & 34 & 8 & 0 & - \\
                inter...convergence.bench & 2530 & 2142 & 1864 & 0.60 & 0.33 & 25 & 23 & 19 & 20 & 20 \\
                b22\_C.bench & 18458 & 14247 & 13936 & 2.32 & 7.22 & 146 & 122 & 23 & 10 & 3 \\
                mult64.bench & 32959 & 17741 & 15882 & 5.13 & 12.20 & 959 & 900 & 0 & - & - \\
                \bottomrule
            \end{tabular}
        \end{center}
        \caption{Comparison of circuit sizes and running times between \texttt{resyn2} and \texttt{resyn2+Simplifier}. The \texttt{r} column corresponds to the size after applying \texttt{resyn2}, and the \texttt{r+s} column corresponds to the size after applying \texttt{resyn2} followed by \texttt{Simplifier}. The iterations column represents the number of subcircuits that were simplified in each iteration of~\texttt{Simplifier}. The times for \texttt{resyn2} and \texttt{Simplifier} are also shown, with the time for \texttt{Simplifier} representing the sum of the time over all iterations.}
        \label{table:iterationStat}
    \end{table}

\subsubsection{BENCH Circuits}
	The \texttt{ABC} framework is primarily optimized for circuits in the $\{\land, \neg\}$ basis.
	As~a result, for the BENCH circuits, shown in Table~\ref{table:statBench}, we report only the percentage of gates removed by our tool.
    (Note that some of~the public datasets are available in~the AIG basis only.)
	Experiments demonstrate that our simplification tool reduces the size of BENCH circuits by an average of 30\%. All circuits used in our experiments are either publicly available or derived from natural encodings of various problems. The results indicate that our tool can efficiently simplify many practical circuits.

\begin{table}[!ht]
	\begin{center}
		\begin{tabular}{lrr}
			\toprule
			Class & Count & Improvement (\%) \\
			\midrule
			Pigeonhole principle & 68 & 46.3 \\
			Even colouring principle & 118 & 14.3 \\
			Clique search problem & 94 & 7.8 \\
			Factorization & 159 & 43.7 \\
			Miter, summation & 20 & 0 \\
			Miter, threshold & 20 & 0.5 \\
			Miter, multiplication & 57 & 69.7 \\
			Miter, sorting & 24 & 69.1 \\
			ISCAS85 & 11 & 36.8 \\
			ITC'99 & 41 & 18.8 \\
			Bristol Fashion MPC Circuits & 22 & 18.1 \\
			\bottomrule
		\end{tabular}
	\end{center}
	\caption{Simplification statistics for various classes of circuits in the BENCH basis.}
	\label{table:statBench}
\end{table}


\subsection{Verifying Simplification}
To~verify that our tool simplifies a~circuit correctly (that~is, outputs
a~circuit computing the same function), we~employed the Combinational Equivalence Checking algorithm~\cite{4110135} implemented within the framework \texttt{ABC} via the \texttt{cec} command.
Using these command, we~verified 
that our tool
simplifies correctly all the circuits
used for the experimental evaluation.

\section{Conclusion}
We developed \texttt{Simplifier}, a new open-source tool for simplifying Boolean circuits.
We have shown that our tool is helpful for low-effort optimization in conjunction with the \texttt{ABC} framework: we evaluated the tool on various Boolean circuits, including industrial and hand-crafted examples, in two popular formats: AIG and BENCH. After applying the state-of-the-art \texttt{ABC} framework, our tool achieved an additional 4\% average reduction in~size for AIG circuits. For BENCH circuits, the tool reduced their size by~an~average of 30\%.

\section*{Acknowledgments}
We~are thankful to~the anonymous reviewers whose numerous comments
helped~us to~improve not only the writing, but also the tool itself.

\bibliographystyle{plain}
\bibliography{references}

\end{document}